\documentclass[pdflatex,sn-mathphys-num]{sn-jnl}%
\usepackage{graphicx}%
\usepackage{multirow}%
\usepackage{amsmath,amssymb,amsfonts}%
\usepackage{amsthm}%
\usepackage{mathrsfs}%
\usepackage[title]{appendix}%
\usepackage{xcolor}%
\usepackage{textcomp}%
\usepackage{manyfoot}%
\usepackage{booktabs}%
\usepackage{algorithm}%
\usepackage{algorithmicx}%
\usepackage{algpseudocode}%
\usepackage{listings}%
\usepackage{tcolorbox}
\usepackage{setspace}
\usepackage{adjustbox}
\usepackage{lineno}
\usepackage{geometry}
\usepackage{afterpage}


\usepackage{caption}
\DeclareCaptionLabelFormat{adja-page}{\hrulefill\\#1 #2 \emph{(previous page)}}
\usepackage{subfig}

\theoremstyle{thmstyleone}%
\newtheorem{theorem}{Theorem}
\theoremstyle{thmstyletwo}%

\theoremstyle{thmstylethree}%

\newtcolorbox[auto counter]{mybox}[1][]{
    title=Box~\thetcbcounter,
    #1 
}

\begin{document}

\title[Measurement noise scaling laws for cellular representation learning]{
\linespread{0.95}\selectfont
  \Large\bfseries
A measurement noise scaling law for cellular representation learning
}

\author*[1, 2, 3]{\fnm{Gokul} \sur{Gowri}}\email{gokulg@mit.edu}
\author[4]{\fnm{Igor} \sur{Sadalski}}
\author[4]{\fnm{Dan} \sur{Raviv}}
\author[1, 2]{\fnm{Peng} \sur{Yin}}
\author[4]{\fnm{Jonathan} \sur{Rosenfeld}}
\author*[1]{\fnm{Allon} \sur{Klein}}\email{allon\_klein@hms.harvard.edu}

\affil[1]{\orgdiv{Department of Systems Biology}, \orgname{Harvard Medical School}, \orgaddress{\city{Boston}, \state{MA}, \country{USA}}}
\affil[2]{\orgdiv{Wyss Institute for Biologically Inspired Engineering}, \orgname{Harvard Medical School}, \orgaddress{\city{Boston}, \state{MA}, \country{USA}}}
\affil[3]{\orgdiv{Laboratory for Information and Decision Systems}, \orgname{Massachusetts Institute of Technology}, \orgaddress{\city{Cambridge}, \state{MA}, \country{USA}}}
\affil[4]{\orgname{Cellular Intelligence}, \orgaddress{\city{Boston}, \state{MA}, \country{USA}}}

\abstract{Large genomic and imaging datasets can be used to train models that learn meaningful representations of cellular systems. Across domains, model performance improves predictably with dataset size and compute budget, providing a basis for allocating data and computation. Scientific data, however, is also limited by noise arising from factors such as molecular undersampling, sequencing errors, and image resolution. By fitting 1,670 representation learning models across three data modalities (gene expression, sequence, and image data), we show that noise defines a distinct axis along which performance improves. Noise scaling follows a logarithmic law. We derive the law from a model of noise propagation, and use it to define noise sensitivity and model capacity as benchmarking metrics. We show that protein sequence representations are noise-robust while single cell transcriptomics models are not, with a Transformer-based model showing greater noise robustness but lower saturating performance than a variational autoencoder model. Noise scaling metrics may support future model evaluation and experimental design.}


\maketitle


\section*{Introduction}\label{sec:intro}

Cellular profiles obtained by single-cell RNA sequencing (scRNA-seq) and high-content imaging now span diverse tissues, developmental stages, disease states, and experimental perturbations \cite{Yao2023-wi, Zhang2025-qv}. These large datasets (collectively $> 10^8$ samples) create opportunities to identify shared cellular states across experimental contexts and predict responses to novel perturbations \cite{Bunne2024-hw, Wang2025-iv}.  To realize these opportunities, representation learning models are used to capture biologically meaningful variation, while filtering out technical nuisance factors \cite{Gunawan2023-gf}. Several deep learning approaches underlie such models to date, including transformer-based architectures, autoencoder-based architectures, and contrastive losses \cite{Theodoris2023-iw, Cui2024-de, Heimberg2024-mc, Richter2024-zs}.

In domains outside of biology including natural language processing, image processing and chemical informatics, large model development has been guided by the study of model scalability. Choices in architecture, data acquisition, and training strategies are guided by deep learning scaling laws, which are empirical relationships that describe how model performance improves with increases in key resources like data, compute, and model parameters \cite{Hestness2017-sq, Rosenfeld2019-wc, Kaplan2020-ug, Hoffmann2022-wc, Bahri2024-dr, Chen2023-oh}.

In biology and other scientific domains, model performance can also be limited by noise in the data used for model training. A few specific data modalities, such as DNA sequence, exist in large repositories with reasonably low error rates ($<10^{-2}$ errors/nucleotide, \cite{Stoler2021-eo}) but the majority of biological data modalities are more prone to measurement noise. scRNA-seq and spatially-resolved transcriptomics, for example, are methods fundamentally limited by the low numbers of mRNA molecules per gene in each cell. Though measurement sensitivity has increased substantially with ongoing development of these methods \cite{Hagemann-Jensen2020-bd}, for many existing technologies the probability of detecting a given mRNA molecule is well below 50\%, and in some cases the detection rate is further decreased by insufficient sequencing depth \cite{Svensson2017-vw}. As a result, measured transcript counts are subject to undersampling noise. Fluorescence microscopy is also prone to noise of different types including shot noise, scattering, out-of-focus light, autofluorescence, and limits in optical and camera resolution \cite{Lichtman2005-da}. 

In contrast to the scaling of model performance with data set size and model size, much less is known about the role of measurement noise on the ability of a model to learn meaningful representations. In textual representation by large language models (LLMs), errors in training data lead to degraded performance, predicted to persist even in the limit of infinite data \cite{Bansal2022-ig}. However, textual data used in LLM training are much less noisy than biological data. As representation models are being developed for diverse biological tasks, understanding how noise alters the learning rate of models could support the model evaluation, and allow rational planning of data-generation campaigns. 

Here, we show evidence for a general and quantitative scaling relationship between measurement noise and model performance across representation learning models of scRNA-seq, spatial transcriptomics, protein sequences, and imaging data. We show that the noise-scaling law can be derived by analogy to additive Gaussian noise channels, and that the resulting analytical form of the law can be used to define model benchmarks and identify noise-sensitivity and saturation that may guide experimental design. We critically evaluate several models.
\section*{Results}\label{sec:results}

\subsection*{A metric for representation-learning model performance} \label{sec:framework}

In deep learning scaling analyses, it is typical to evaluate the quality of models directly by evaluating their loss on held-out validation data \cite{Hestness2017-sq, Rosenfeld2019-wc, Kaplan2020-ug, Bahri2024-dr}. 
However, model loss is not comparable between models with different loss functions, or even for a single model applied to data with different statistical properties \cite{Brandfonbrener2024-om} such as different noise properties.
Therefore, to study the effect of noise on representation learning model performance, we introduced an alternative approach to measuring representation quality, by estimating the mutual information (MI) between the representations learnt by a model and some information about each sample that remains hidden until after learning is completed (\textbf{Fig.~\ref{fig:schematic}}). Formally, this approach is a generalization of linear probing \cite{Belinkov2022-sw, Pimentel2020-po}, which can be viewed as estimating MI between a representation and a classification label.  Our generalization uses a neural network-based estimator of MI that accommodates high-dimensional and continuous auxiliary signals \cite{Gowri2024-wz}. This approach provides a performance metric that is comparable between model types and noise levels in a given data set. 

We first evaluated representation model performance for four single-cell transcriptomic data sets, each of which provides an additional auxiliary signal as follows: 
\begin{enumerate}
    \item \textit{Developmental time} of $\sim10^7$ cells profiled by scRNA-seq across mouse development, where developmental time is quantified by embryonic stage \cite{Qiu2024-cq}.
    \item \textit{Surface protein abundances} of $\sim 10^5$ peripheral mononuclear blood cells (PBMCs) measured by an antibody panel through CITE-seq \cite{Hao2021-yu}.
    \item \textit{Transcriptional profile of a clonally related cell} in $\sim 10^5$ mouse hematopoietic stem cells measured using lineage-traced scRNA-seq \cite{Weinreb2020-uk}. 
    \item \textit{Transcriptional profile of a spatially adjacent cell} in a coronal mouse brain section of $\sim 10^5$ cells measured using MERFISH \cite{Vizgen2021-fo}.
\end{enumerate}

For each of these, we evaluated two linear baseline models: random projection and dimensionality reduction by principal component analysis (PCA), and two modern generative models with distinct architectures: single cell variational inference (scVI) \cite{Lopez2018-zz} and Geneformer \cite{Theodoris2023-iw}. scVI is a variational autoencoder designed to compress high-dimensional gene expression information into a low-dimensional latent space, with an explicit treatment of gene expression noise. Geneformer instead uses a Transformer-based language model that maps gene expression vectors to sequences by ordering gene-specific tokens based on expression level. Our implementations of scVI and Geneformer have between 1.3-116.8 million and 1.9-13 million parameters respectively (depending on gene number in each dataset), although we note that parameter counts across different model families are not directly comparable. We trained Geneformer across multiple GPUs using DeepSpeed \cite{Rasley2020-iq}. Model implementation and data preprocessing details  can be found in the {\bf Supplemental Text} (sections~\ref{app:details_data}, \ref{app:details_model}).

In all cases, to facilitate consistent comparisons, we learned representations on one data subset, and then evaluated performance in a separate, fixed held-out subset. This approach ensures that observed differences in mutual information are attributable to variations in the representations themselves, rather than estimation artifacts.

\begin{figure}
    \centering
    \includegraphics[width=\linewidth]{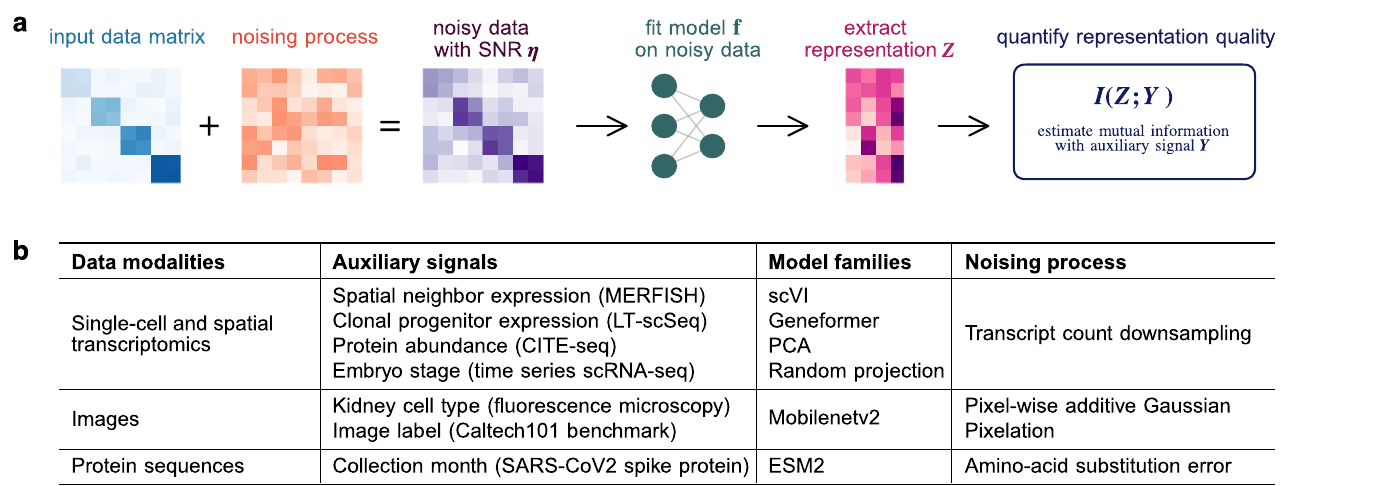} \\ \ \\
        \caption{ \textbf{Workflow and data sets for studying noise in biological representation learning.}
    \textbf{(a)} The workflow used to evaluate representation model quality as a function of noise, with model performance quantified by auxiliary mutual information (see text). The process is repeated across many noise levels (signal-to-noise ratios $\eta$) to generate empirical noise-scaling curves.
    \textbf{(b)} Summary of the auxiliary signals, model families and noising processes used for each of the three data modalities in this study.
    }
    \label{fig:schematic}
\end{figure}
\subsection*{Cell number scaling for cellular representations}\label{sec:cellnumber}
As a baseline for understanding the impact of noise on model learning, we first tested whether auxiliary-MI performance, $I$, shows expected scaling behavior with the number of samples $N$ (here, the number of single cells) used in training. In image and language models, performance scales as a power of sample number \cite{Bahri2024-dr}. We indeed found that $I$ is well-described by a saturating power-law across all neural network-based models and auxiliary tasks, $ I(N) = I_{\infty} - ({N}/{N_\textrm{sat}})^{-s} $,  where the parameters $I_\infty, N_{\textrm{sat}}, s$ characterize how each model learns from new data (fit coefficient of determination $R^2=0.90\pm 0.02$ across $n=80$ models comprising 8 scaling curves). The fits are shown collectively across models and datasets in {\bf Figs.~\ref{fig:bigfig}a, b}, with parameter values and model comparisons in {\bf Supplemental Fig.~\ref{fig:cellnumberparams}}. The differences in scaling between models is in agreement with prior work ({\bf Supplemental Text~\ref{app:cellnumberdetails}} and \cite{DenAdel2024-ue}), and establishes auxiliary-MI as suitable for studying the impact of noise on model performance.

\afterpage{%
  \clearpage
  \newgeometry{top=0.3in, bottom=0.3in, left=1in, right=1in}
\begin{figure}[p]
    \centering
    \captionsetup{labelformat=empty}
    \centering
  \begin{adjustbox}{width=1\textwidth,center}
    \includegraphics[width=\textwidth]{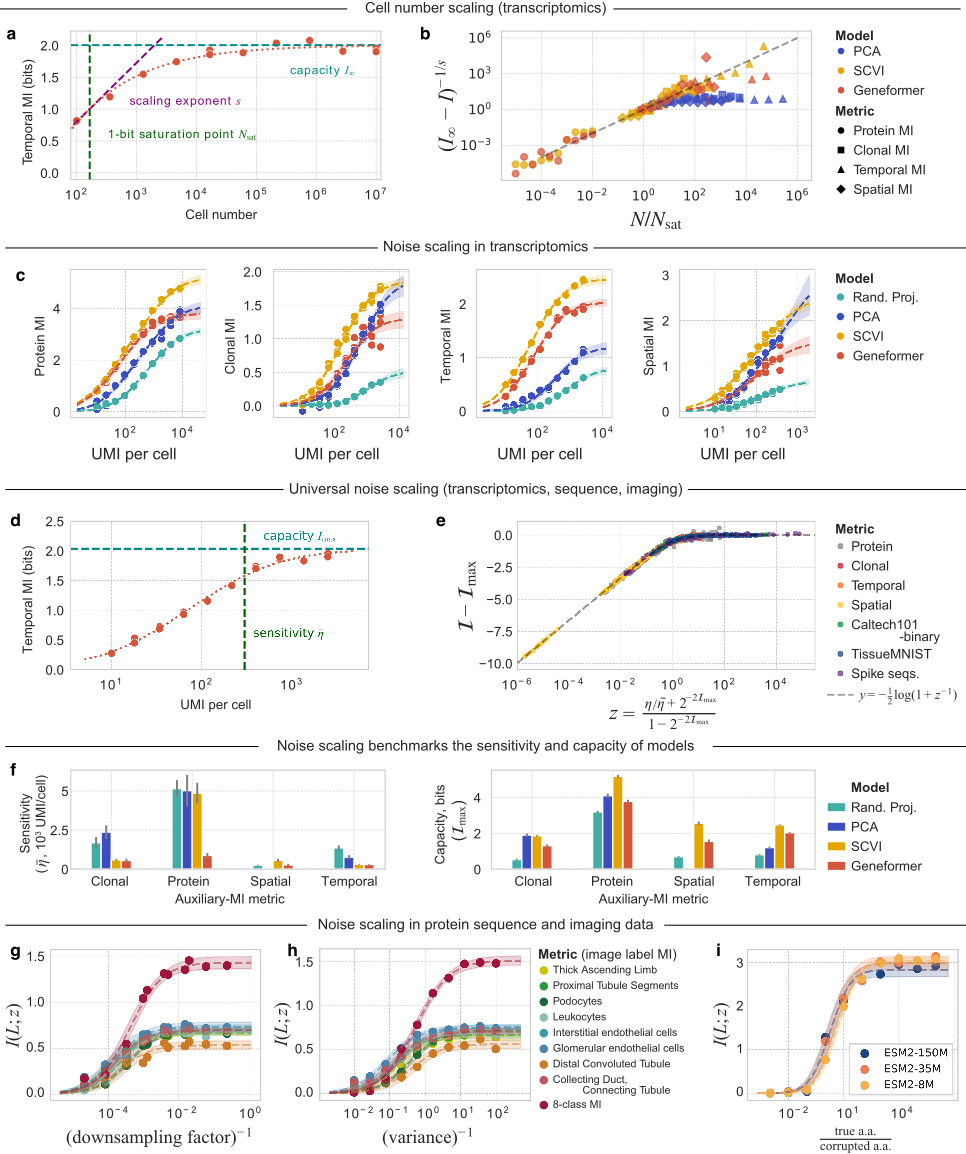}
  \end{adjustbox}
    \caption{\small \textbf{Fig. 2. Noise scaling in representation learning.} 
    \textbf{(a,b)} Representation model performance shows saturating power-law scaling with cell number. In (a), one example is shown of the Geneformer model applied to a 12 million cell single-cell transcriptomic data set \cite{Qiu2024-cq}. In (b), the data is replotted with rescaled axes and extended to three model families and four datasets (see Fig.~1b), to show common scaling behavior. $n=120$ models, across 3 tasks and 4 metrics.
    \textbf{(c)} Representation model performance as a function of noise, by progressive downsampling of transcripts detected per cell. Panels show different datasets. UMI=Unique Molecular Identifier. Dashed curves show fits to Eq.~(1), with $2\sigma$ confidence intervals shown by the shaded regions.
    \textbf{(d)} Annotation of model capacity and noise-robustness parameters $I_{\max},\ \bar{\eta}$ on a typical noise scaling plot.
    \textbf{(e)} Scaling plot showing collapse of 193 different noise scaling curves derived from measurements of 1,670 fit models spanning transcriptomic data, image data, and sequence data. Theory curves fit with $R^2= 0.979 \pm 0.003$ (mean$\pm$SD).
    \textbf{(f)} Benchmarking models by noise-sensitivity and saturating performance. Parameters for PCA on the spatial metric are unconstrained and are omitted from the plot.
    \textbf{(g-i)} The underlying representation model performance plots from (e), for kidney cortex cell microscopy image MobileNet representations after progressive noising by downsampling (g) and additive noise (h); and for COVID19 protein spike sequence ESM2 representations with progressive amino acid substitutions (i). Confidence bands show $2\sigma$ interval. $L$ denotes image/sequence label, and $z$ denotes representation.
    }
    \label{fig:bigfig}
\end{figure}
    \restoregeometry
    \clearpage
}

\subsection*{Noise scaling for cellular representations}

Although deep learning models are often thought to be strong denoisers \cite{Bengio2012-zx}, the degree to which cellular representation learning models are robust to noise in their training data is unknown. A noise-robust model would exhibit a regime in which the informativity of learned representations remains stable despite increasing noise levels. We evaluated the extent to which models are noise robust by simulating increasing measurement noise through downsampling observed transcript counts, and subsequently evaluating the quality of the learned representations using auxiliary-MI performance.

The dependence of model performance $I$ on the degree of downsampling noise is shown in {\bf Fig. \ref{fig:bigfig}c}. As expected, reducing the depth per cell degrades the performance of all models, across all datasets. Of note, no model or dataset exhibits large regimes of noise robustness. Instead, many of the measured performance curves are sigmoidal, indicating only limited robustness at the transcript levels present in the original datasets before performance steadily deteriorates ({\bf Fig. \ref{fig:bigfig}c}). A subset of the curves (e.g., PCA representations for spatial information) show a `hockey-stick' shape, indicating negligible robustness to downsampling noise, even at full transcript levels.

The smooth, sigmoidal performance curves as a function of downsampling noise suggest that a simple analytical relationship might capture how measurement noise constrains biological representation learning. Such a relationship would be valuable for experimental design, enabling principled allocation of sequencing depth and cell numbers in the same way that neural scaling laws guide resource decisions in large-scale machine learning. By inspection, it is evidently not a power-law, unlike previous neural scaling laws \cite{Bahri2024-dr}.

To define the noise–performance relationship, we turned to a simple model of information loss in noisy communication channels (see {\bf Box~\ref{box:theory}}). By extending established information-theoretic results \cite{Guo2005-bm}, we derived the following relationship between the signal-to-noise ratio of a measurement, $\eta=\, $SNR, and the mutual information preserved about an underlying external variable (see \textbf{Box~\ref{box:theory}}, \textbf{Supplemental Text \ref{app:proof}}):

\vskip 0.1in
\begin{equation}\label{eqn:scaling_general}
    \mathcal{I}(\eta) = \mathcal{I}_{\max} - \frac{1}{2} \log \frac{\eta/\bar{\eta} + 1}{\eta/\bar{\eta} + 2^{-2\mathcal{I}_{\max}}},
\end{equation}
\vskip 0.1in
where $\mathcal{I}_{\max}$ is the maximal information that can be extracted from a noiseless measurement at a fixed sample size, and $\bar{\eta}$ serves as a measure of noise robustness (specifically, the signal-to-noise ratio at which a model can gain at most $1/2$ a bit of information by increasing measurement sensitivity). These parameters are annotated on an empirical curve for Geneformer and temporal information in the mouse embryogenesis dataset in \textbf{Fig.~\ref{fig:bigfig}d}.
To connect this general relationship to cellular measurements, we note that the signal-to-noise ratio introduced by molecular undersampling follows Poisson statistics (see \textbf{Supplemental Text~\ref{app:umi-to-snr}}), $\eta = \mathrm{CV}^{-2} \propto \mathrm{UMI}$ \cite{Klein2015-jz} (CV=Coefficient of Variation), and $\bar{\eta}$ then takes on units of UMIs per cell.

In \textbf{Fig.~\ref{fig:bigfig}c}, the theoretical curves defined by Eq.~(\ref{eqn:scaling_general}) (dotted lines) closely match the empirical performance curves (scatter points) across models and datasets. The noise-scaling relationship holds across model architectures and across single-cell datasets spanning nearly five orders of magnitude in sample size. To illustrate generality of the relationship, when the 160 empirical measured curves are rescaled by their fitted $\mathcal{I}_{\max}$ and $\bar{\eta}$ values,  all models from all four model families collapse onto a single universal relationship (\textbf{Fig.~\ref{fig:bigfig}e}). 

The fitted noise-scaling parameters from Eq.~\ref{eqn:scaling_general} provide a compact summary of the noise-robustness of a model. In particular, $\bar{\eta}$ reflects each model’s effective noise tolerance, while $\mathcal{I}_{\max}$ captures its asymptotic capacity in the absence of measurement noise. For a given auxiliary task, models that combine low $\bar{\eta}$ with high $\mathcal{I}_{\max}$ are therefore preferred.

In \textbf{Fig.~\ref{fig:bigfig}f}, we compare inferred $\bar{\eta}$ and $\mathcal{I}_{\max}$ values across models. Geneformer consistently shows the greatest robustness to noise: across all tasks, it approaches within 0.5 bits of its asymptotic performance at fewer than $1{,}000$~UMI per cell. scVI displays similarly low noise sensitivity for three of the four tasks, but in the protein-abundance task it becomes noise-sensitized at $\sim 4{,}000$~UMI per cell. PCA, by contrast, shows far greater sensitivity to noise, with $\bar{\eta}$ values 2.5–12.8-fold larger than those of Geneformer, consistent with the limited denoising capacity of linear methods.

Despite its robustness to noise, Geneformer is not a strong model in terms of its capacity. Across all tasks, its capacity, $\mathcal{I}_{\max}$, is lower than those of scVI by 0.4–1.4~bits -- corresponding to approximately halving the complexity of the captured signal. This difference in performance is not only in its asymptotic capacity, but also at the noise level present in the datasets ({\bf Fig. \ref{fig:bigfig}c}).  Thus, representations learnt by scVI ultimately capture more information in the limit of low noise. It is possible that other models may simultaneously show noise robustness and higher information capacity.

\begin{mybox}[colback=blue!5!white,colframe=blue!30!black,
title=\textbf{Box 1:} \sf A model of noise scaling in representation learning,label=box:theory]
\begin{spacing}{1}
{\sf 
The empirical noise–performance curves in {Fig.~\ref{fig:bigfig}d} suggest that a simple theoretical relationship underlies how measurement noise limits the information extractable by representation models.  A classical setting in which such limits are analytically tractable is a Gaussian noise channel, where both the signal and the noise are modeled as Gaussian random variables. Although simplified, this framework captures the essential effect of diminishing returns: as measurement quality improves, each additional increment in signal-to-noise ratio (SNR) conveys progressively less new information.  We use it to derive the scaling form in Eq.~(\ref{eqn:scaling_general}).


Let $X, Y$ be multivariate Gaussian random vectors representing the system state and an auxiliary signal,  and let $Z$ be a noisy measurement of $X$ with SNR~$\eta$:
\[
Y\sim\mathcal{N}(\mu_Y,\Sigma_Y), \qquad
X = Y + \mathcal{N}(0,\Sigma_U), \qquad
Z = \sqrt{\eta} X + \mathcal{N}(0,I_n).
\]

Here, $\eta$ corresponds to the power-ratio SNR when signals are normalized such that $\mathbb{E}[X^2]=1$ (following the convention of \citet{Guo2005-bm}). The mutual information between $Y$ and $Z$ -- the amount of auxiliary signal retained after measurement -- follows from a standard expression for Gaussian vector noise channels \cite{Guo2005-bm, Polyanskiy2024-nx} (proof in Appendix~\ref{app:proof}):
\[
I(Y;Z) = \tfrac{1}{2}
\log \frac{\det(\Sigma_Y+\Sigma_U+\eta^{-1}I_n)}
               {\det(\Sigma_U+\eta^{-1}I_n)}.
\]

For the scalar case ($n=1$), where $\Sigma_Y=\sigma_Y^2$ and 
$\Sigma_U=\sigma_U^2$,
\begin{equation}\label{eqn:toy}
I(Y;Z) = \tfrac{1}{2}\log
\frac{\eta(\sigma_Y^2+\sigma_U^2)+1}
     {1+\sigma_U^2\eta}.
\end{equation}

Two characteristic quantities govern this scaling:
\[
\mathcal{I}_{\max} = \lim_{\eta\to\infty} I(Y;Z)
= \tfrac{1}{2}\log\frac{\sigma_Y^2+\sigma_U^2}{\sigma_U^2},
\]

which reflects the maximal information achievable with a noiseless measurement, and 
$\bar{\eta} = 1/\sigma_U^2$, an effective noise scale. Substituting $\mathcal{I}_{\max}$ and $\bar{\eta}$ into $I(Y;Z)$ recovers precisely the 
empirical noise-scaling relationship of Eq.~\ref{eqn:scaling_general}.  
Despite its simplicity, this model captures the universal shape of the 
performance–noise curves observed across datasets and architectures.
} 
\end{spacing}

\end{mybox}

\clearpage
\subsection*{Generalization of noise scaling}\label{generalization}
The scaling law (Eq.~\ref{eqn:scaling_general}) depends only on the signal-to-noise ratio $\eta$, and a model that explains this law ({\bf Box~\ref{box:theory}}) is not specific to transcriptomic data. To test whether the same scaling relationship generalizes to other models, we examined noise–performance relationships in image representation models, and then in protein sequence models. 

For image representation, we used MobileNetV2, a lightweight convolutional architecture designed for image classification \cite{Sandler2018-bt}. We trained and evaluated this model on two different image datasets (1) a 5-class subset of the Caltech101 dataset \cite{Fei-Fei2004-ux}, consisting of $240\times240$ pixel images with 2,707 total images and (2) a fluorescence microscopy dataset of 236,386 human kidney cortex cells annotated with one of eight cell type labels \cite{Yang2023-eg}. Images were perturbed with two distinct forms of degradation: additive Gaussian noise and reduced spatial resolution. Pixel-wise Gaussian noise is common in imaging measurements \cite{Da_Costa2016-pm}. We then used auxiliary-MI to evaluate model performance under both forms of image noise, here measuring the mutual information between the learned representations and the true image labels. We trained the MobileNetV2 models and computed the auxiliary-MI between predicted and true labels on held-out images, assessing performance for two tasks: classification of all class labels, as well as multiple one-vs-all problems.  We introduced Gaussian noise with $\eta = 1/\sigma_N^2$, where $\sigma_N$ is the noise standard deviation, while resolution degradation was introduced by averaging local pixel neighborhoods, with $\eta = 1/f$ for downsampling factor $f$. For both types of noise, we found that Eq.~\ref{eqn:scaling_general} accurately reproduced the observed noise–performance curves for all classification tasks ({\bf Figs.~\ref{fig:bigfig}g,h}, $R^2=0.984 \pm 0.004$ for Caltech101 and $R^2 = 0.979 \pm 0.005$ for kidney cortex models).

A similar pattern emerged in representation learning models of protein sequence. We finetuned ESM2 models (8M, 35M, and 150M parameter variants) \cite{Rives2021-ws} on a set of $\sim$63,000 SARS-CoV-2 spike protein sequences spanning the course of the pandemic \cite{Shu2017-au}, after introducing controlled levels of amino-acid substitution to simulate increasingly corrupted measurements. We then quantified auxiliary-MI between the learned representations and the collection date of each sequence, measured as number of months since the pandemic outbreak in January 2020 (with sequences up to April 2025). Despite the discrete and highly structured nature of protein sequences, the resulting noise–performance curves again closely match the form predicted by Eq.~\ref{eqn:scaling_general}, with increasing substitution rates driving systematic and predictable declines in mutual information ({\bf Fig.~\ref{fig:bigfig}i}). 

The shared behavior of these protein sequence and imaging models is demonstrated by collapsing the 33 additional image and sequence curves by appropriate rescaling in {\bf Fig.~\ref{fig:bigfig}e}.  Together, this analysis adds to the evidence that noise in training data is a systematic determinant of representation quality -- one that can be modeled alongside sample size when characterizing learning behavior.
\subsection*{Noise scaling and experimental design}

Measurement noise scaling laws can be used to determine the data quality or sample quantity required to achieve a specified level of representation performance. The parameter $\bar{\eta}$ from Eq.~\ref{eqn:scaling_general} directly reports the measurement sensitivity at which model performance reaches within 0.5 bits (or approximately 70\%) of its asymptotic value. More generally, inverting Eq.~\ref{eqn:scaling_general} yields a function $\eta(\mathcal{I})$ that predicts the minimum signal-to-noise ratio needed for a model to achieve a desired information content with respect to a given auxiliary variable (see \textbf{Supplemental Text~\ref{sec:app_UMI90}}, Eqn.~\ref{eqn:umi90}).

For transcriptomic data, $\eta$ is proportional to the total UMIs per cell, enabling an estimate of the sequencing depth needed for a representation to reach, for example, 90\% of its maximum informativity. These depth requirements (${ \rm UMI}{90}$) are reported for all model–task pairs in \textbf{Supplemental Table~\ref{tab:u90}}. Several clear patterns emerge. Geneformer consistently operates above its ${\rm UMI}{90}$ on all datasets examined, indicating that its performance is already near its asymptotic limit under typical sequencing depths. In contrast, ${\rm UMI}{90}$ for scVI exceeds the observed UMI counts for protein abundance, spatial information, and clonal information tasks—suggesting that these tasks remain sensitivity-limited and would benefit from deeper sequencing. 
These examples illustrate how noise-scaling relationships can guide the choice of models and allocation of sequencing depth across tasks with different intrinsic difficulty.

The same experimental-design questions can be asked for image and protein-sequence representation learning. For image classification ({\bf Fig.~\ref{fig:bigfig}g,h}), the fitted $\bar{\eta}$ values indicate, for example, the image resolution necessary for a given task. For kidney cell type annotation task, under pixelation noise, $\eta_{90}\approx 3\cdot 10^{-3}$ corresponding to an effective resolution threshold of $\sim 15 \times 15$ pixels: images downsampled beyond this point lose more than $\sim 0.15$ bits of label information, corresponding to $10\%$ information loss. Certain classes (e.g., Podocyte) exhibited a steeper performance decay (larger $\bar{\eta}$), indicating that their recognition relies on higher-resolution features. 
For protein sequence, auxiliary-MI remained stable up to substitution rates of approximately 1 in 1,000 amino acids -- a noise level well above what is typical in modern sequencing. This is consistent with DNA and protein sequence models to date being able to largely ignore measurement noise.
\section*{Discussion}

Noise in training data inevitably affects model performance, but it has remained unclear whether there exist predictable, quantitative rules governing how representation quality degrades as noise increases. Across single-cell transcriptomics, imaging, and protein sequence data, we find that auxiliary-task performance follows a characteristic sigmoidal scaling curve. Models retain robust performance above a modality-specific noise threshold, after which representation quality declines approximately logarithmically with increasing noise. A simple information-theoretic model captures this relationship and recovers the empirical scaling form observed across more than $10^3$ learned representations. These results suggest that predictable noise-dependent learning curves may be a common feature across diverse biological data modalities.

This work provides practical guidance for designing and evaluating biological representation models. The fitted scaling parameters $\bar{\eta}$ and $\mathcal{I}_{\max}$ jointly characterize model behavior: $\bar{\eta}$ reflects noise robustness, while $\mathcal{I}_{\max}$ represents the maximal task-relevant information that a model can encode. Nonlinear models such as Geneformer and scVI exhibit substantially greater robustness to measurement noise than PCA, consistent with the expectation that nonlinear architectures more effectively denoise sparse molecular measurements. However, robustness alone is insufficient. Geneformer, despite its stability under noise, often attains a relatively low $\mathcal{I}_{\max}$, capturing less auxiliary information than scVI and, for certain tasks, even linear baselines. These results emphasize that noise robustness and representational capacity must be jointly optimized in model design.

Noise scaling has implications for experimental design, particularly for large-scale single-cell profiling. Our analysis shows that some tasks, such as our test tasks of predicting surface-protein or spatial information from scRNA-seq, remain sensitivity-limited even in current datasets. These tasks would benefit substantially from higher per-cell transcript counts. Conversely, for tasks such as predicting developmental stage in the mouse embryo atlas, existing sequencing depth is already sufficient to approach the representational limit. These distinctions highlight that improvements in measurement quality, rather than cell number alone, may be the most impactful direction for next-generation atlases and molecular profiling initiatives.

More broadly, considering measurement noise as an additional scaling axis suggests a more complete picture of representation learning in `measurement-bound' fields such as biology, complementing the well-established roles of dataset and model size in neural scaling. Noise imposes a predictable, quantifiable constraint that can be analytically modeled and experimentally addressed with suitable trade-offs. Often, there may be a trade-off between sample number and measurement sensitivity. One can design assays that sit on or near the optimal learning curve for a given task.


Several questions remain. First, an open theoretical question is to understand the origin of the scaling law. The model we introduce here ({\bf Box~\ref{box:theory}}) is exact for scalar Gaussian channels, yet it fits high-dimensional biological data surprisingly well. Understanding why this is the case, and under what conditions noise scaling breaks down, represents a theoretical direction. Second, we have still only demonstrated noise scaling in a small number of modeling tasks. Third, even for the tasks at hand, we have only evaluated a small number of model architectures. The high cost of training foundation models makes it impractical for us to evaluate additional models. It is possible that finetuning of pre-trained models may provide a faithful probe of noise-tolerances of a model, allowing systematic evaluation of additional models. Finally, our analysis has treated measurement noise and dataset size separately, and has not considered model size or compute budget; developing a joint scaling law that unifies multiple resource axes would further clarify how to allocate resources to build predictive models of high-dimensional biological systems.

In sum, our findings suggest that measurement noise is a predictable and actionable determinant of representation model performance, one that can be optimized alongside dataset size to guide both model development and experimental design across biological modalities.

\backmatter

\bmhead{Code availability}

All code necessary to reproduce the results in this work is provided at \href{https://github.com/g-kl/noise-scaling/}{this GitHub repository} (\verb|g-kl/noise-scaling|).

\bmhead{Acknowledgments}

This work is supported by funding from NIH Pioneer Award DP1GM133052, R01HG012926 to P.Y., and Molecular Robotics Initiative at the Wyss Institute. A.M.K. acknowledges support of an Edward Mallinckrodt Jr. Scholar Award. G.G. acknowledges support from the Tayebati Postdoctoral Fellowship Program.

\clearpage

\begin{appendices}
\renewcommand{\appendixname}{}
\renewcommand{\thesection}{S.\Roman{section}}

\setcounter{figure}{0}
\renewcommand{\thefigure}{S\arabic{figure}}
\setcounter{table}{0}
\renewcommand{\thetable}{S\arabic{table}}
\setcounter{equation}{0}
\renewcommand{\theequation}{S\arabic{equation}} 

\begin{center}
\section*{Supplemental Text}
\end{center}
\addcontentsline{toc}{section}{Supplemental Text}

\section{Data preprocessing methods}\label{app:details_data}

In this section, we summarize key details of our data preprocessing methods for each of the data sets. We also provide an annotated codebase in the supplemental files. 

\subsection{MERFISH mouse brain dataset for spatial information probing}

We use the 67,821 single cell transcriptomes measured in coronal section 1 of replicate 1 in the Vizgen mouse brain data release \cite{Vizgen2021-fo}. We remove ``blank'' measurements from the dataset, leaving 649 measurement dimensions. We define cell location by the center coordinate of the cell segmentation mask (which is provided in the metadata of the dataset). We construct a paired dataset of neighboring cells by randomly selecting one of the 5 nearest cells as the neighbor pair for each cell in the dataset. Information probing then measures the information each cell representation contains about its neighbor pair.

\subsection{LARRY hematopoiesis dataset for clonal information probing}

We use the \textit{in vitro} differentiation dataset from \citet{Weinreb2020-uk}. We pair clonally related cells as follows. We first separate the dataset into cells profiled at early timepoints (day 2 and day 4), and final day 6 timepoint. Then, we subset the dataset for cells whose clonal barcodes appear in both early and late timepoints. For each remaining clone, we randomly select a cell from the early timepoint and pair it with a randomly selected cell from a late timepoint. Information probing then measures the information each cell representation contains about its clonally related pair.

\subsection{CITE-seq PBMC dataset for protein abundance information probing}

We use the CITE-seq PBMC dataset from \citet{Hao2021-yu}, with the count matrix preprocessed and distributed by \verb|scvi-tools| \cite{Virshup2023-vh}. Information probing measures the mutual information between transcriptome representations and protein abundance vectors.

\subsection{Caltech101}

We use the Caltech101 \cite{Fei-Fei2004-ux} as distributed by PyTorch \cite{Paszke2019-bf}. We rescale pixel intensity values to \verb|[-1, 1]|, and crop images to $240 \times 240$ pixels. We then select the 5 classes with the largest number of images and subset only images from those 5 classes. This leaves a total of 2707 images. To downsample resolution by factor $f$, we tile the image in $240/\sqrt{f} \times 240/\sqrt{f}$ and each pixel is reassigned with the mean pixel value within its respective tile, in effect pixelating the image. To add Gaussian noise, we sample a $240 \times 240$ matrix i.i.d. Gaussians with 0 mean and specified variance for each image and add it to the pixel values.

\subsection{Kidney Cortex}

Kidney tissue nuclear stain images were obtained from the MedMNIST dataset \cite{Yang2023-eg} at size $224\times224$ pixels and were preprocessed through a standardized transformation pipeline. First, pixel values were normalized to the range [0, 1]. As the images contain a single channel (DAPI), channels were replicated to create 3-channel RGB-format images by repeating the single channel three times. Then, the noising process -- either pixel-wise additive Gaussian noise or patch-based pooling (as in the Caltech101 experiments) to simulate pixelation -- was applied. 

\subsection{SARS-CoV2 sequences}

We obtain SARS-CoV-2 spike protein sequences from the GISAID database \cite{Shu2017-au} until the month of 04/2025. For practicality, we use a subsample of sequences. As sequences from certain months (e.g., early 2021) are highly overrepresented, we chose not to uniformly randomly subsample the sequences. Instead, we sampled with a cap of 1000 sequences per collection month. This results in a total of 63,374 sequences across 71 months. We then randomly split this data into 75\% training sequences and 25\% test sequences. Then for  each noise level, we randomly replace amino acids with a new amino acid uniformly sampled from the alphabet with a rate according to the noise level. We then tokenize the noised sequence using the ESM2 tokenizer distributed by HuggingFace \cite{Wolf2020-jb} (\verb|esm2_t6_8M_UR50D|).

\section{Model implementation details} \label{app:details_model}

Below we summarize the implementation details of the models we study in this work.

\subsection{Random projection implementation}

Random projection was implemented using \texttt{sklearn.random\_projection.GaussianRandomProjection}\cite{Pedregosa2011-ed} with 16 components and random state 42. The method operates on gene expression without normalization (i.e., without rescaling or highly variable gene selection).

\subsection{PCA implementation}

We first further preprocess the count matrix by rescaling counts to $10^4$ per cell, and log transforming. We then standardize each gene to zero mean and unit variance and subset to the 750 highly variable genes. Next, we compute principal components using the truncated SVD method implemented in \verb|sklearn| \cite{Pedregosa2011-ed}.

\subsection{Geneformer}

Geneformer was implemented using the original repository \cite{Theodoris2023-iw} as a BERT-based transformer (\verb|BertForMaskedLM|). Training was performed using the HuggingFace Trainer API with length-grouped batching to improve efficiency by grouping cells with similar numbers of expressed genes. To ensure consistent training across diverse dataset sizes, epochs scale dynamically as $\max(1, \lfloor 10 \times (10,000,000 / \text{dataset\_size}) \rfloor)$, ensuring smaller datasets receive proportionally more training while larger datasets train for fewer epochs, thus maintaining a constant effective training budget. Early stopping prevents overfitting by monitoring validation loss, and the model uses masked language modeling with 15\% token masking. Parameter count scales linearly with vocabulary size (number of genes) due to the embedding layer and language model head. Measured parameter counts across datasets range from 1.9M to 13.5M depending on the number of genes in each dataset (\textbf{Supplemental Table~\ref{tab:parameter_counts}}). Complete hyperparameters are provided in \textbf{Supplemental Table~\ref{tab:geneformer_params}}.


\textbf{Supplemental Fig. ~\ref{fig:geneformer_training_loss}} shows the training and validation loss curves for the Geneformer model trained on the developmental task using the full dataset of 10 million cells at full quality level. The model demonstrates stable convergence with both the training and validation losses decreasing monotonically, and the validation loss closely tracks the training loss, indicating good generalization without overfitting. Early stopping was triggered based on validation loss.

\begin{figure}[!htbp]
\centering
\includegraphics[width=0.4\textwidth]{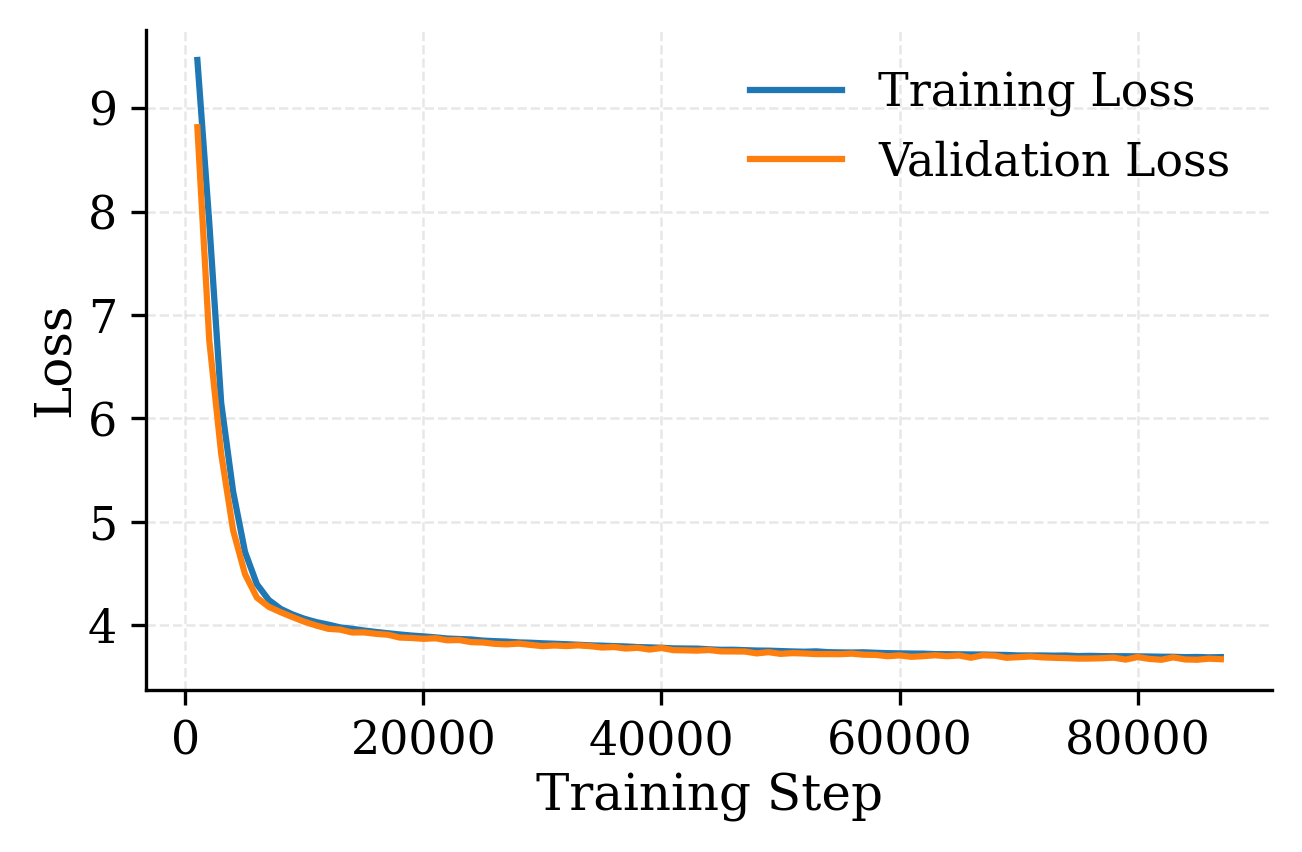}
\caption{\textbf{Loss curves for a representative large-scale model.} Training and validation loss curves for a Geneformer trained on the developmental dataset \cite{Qiu2024-cq} with $\sim10^7$ cells at full transcript count (no artificial noising). The model shows stable convergence.}
\label{fig:geneformer_training_loss}
\end{figure}

\begin{table}[!htbp]
\centering
\caption{Hyperparameters and Implementation Details for Geneformer}
\label{tab:geneformer_params}
\small
\begin{tabular}{ll}
\toprule
\textbf{Parameter} & \textbf{Value} \\
\midrule
\multicolumn{2}{l}{\textit{Model Architecture}} \\
Number of layers & 3 \\
Hidden dimension & 256 \\
Attention heads & 4 \\
Intermediate size & 512 \\
Max sequence length & 512 tokens \\
Activation function & ReLU \\
Attention dropout & 0.02 \\
Hidden dropout & 0.02 \\
Initializer range & 0.02 \\
Layer norm epsilon & $1 \times 10^{-12}$ \\
\midrule
\multicolumn{2}{l}{\textit{Training Hyperparameters}} \\
Learning rate & $1 \times 10^{-3}$ \\
Optimizer & AdamW \\
Weight decay & 0.001 \\
Train batch size (per device) & 64 \\
Eval batch size (per device) & 100 \\
\midrule
\multicolumn{2}{l}{\textit{Learning Rate Schedule}} \\
LR scheduler type & Linear \\
Warmup steps & 5,000 \\
\midrule
\multicolumn{2}{l}{\textit{Training Configuration}} \\
Evaluation strategy & Steps \\
Evaluation steps & 1,000 \\
Max epochs & Dynamic (depends on dataset size) \\
\midrule
\multicolumn{2}{l}{\textit{Early Stopping}} \\
Early stopping & Based on validation loss \\
Patience & 5 steps \\
\bottomrule
\end{tabular}
\end{table}

\subsection{scVI}

We use the scVI \cite{Lopez2018-zz} implementation distributed by \verb|scverse| \cite{Virshup2023-vh} (\texttt{scvi.model.SCVI}) with a shallow architecture. The model uses gene-specific dispersion and zero-inflated negative binomial (ZINB) likelihood. Parameter count scales with the number of input genes due to gene-specific parameters in the decoder, and range from 1.3M to 116M (see \textbf{Supplemental Table~\ref{tab:parameter_counts}}). Complete hyperparameters are provided in \textbf{Supplemental Table~\ref{tab:app_scvi_details}}.

Similar to Geneformer, epochs scale dynamically to maintain consistent number of steps across dataset sizes. KL divergence warmup is used to prevent the model from collapsing to the prior early in training. SCVI operates on raw count data without further preprocessing. We do not select for highly variable genes.

\begin{table}[!htbp]
\centering
\caption{Hyperparameters and Implementation Details for SCVI}
\label{tab:app_scvi_details}
\small
\begin{tabular}{ll}
\toprule
\textbf{Parameter} & \textbf{Value} \\
\midrule
\multicolumn{2}{l}{\textit{Model Architecture}} \\
Hidden size & 512 \\
Latent dimension & 16 \\
Number of hidden layers & 1 \\
Dropout rate & 0.1 \\
Dispersion & Gene-specific \\
Gene likelihood & Zero-inflated negative binomial (ZINB) \\
Latent distribution & Normal \\
\midrule
\multicolumn{2}{l}{\textit{Training Hyperparameters}} \\
Learning rate & $1 \times 10^{-3}$ \\
Optimizer & Adam \\
Weight decay & $1 \times 10^{-6}$ \\
Adam epsilon & 0.01 \\
Batch size & 512 \\
\midrule
\multicolumn{2}{l}{\textit{Learning Rate Schedule}} \\
LR scheduler & Reduce on plateau \\
LR scheduler metric & Validation ELBO \\
Minimum LR & $1 \times 10^{-6}$ \\
\midrule
\multicolumn{2}{l}{\textit{KL Annealing}} \\
KL warmup epochs & 1 \\
Max KL weight & 1.0 \\
Min KL weight & 0.0 \\
\midrule
\multicolumn{2}{l}{\textit{Training Configuration}} \\
Train/Val split & 80\% / 20\% \\
Shuffle split & True \\
Max epochs & Dynamic (depends on dataset size) \\
\midrule
\multicolumn{2}{l}{\textit{Early Stopping}} \\
Early stopping & Based on validation ELBO \\
Patience & 5 epochs \\
Min delta & 0.01 \\
\bottomrule
\end{tabular}
\end{table}

\begin{table}[!htbp]
\centering
\caption{Parameter counts for SCVI and Geneformer models across different datasets (millions of parameters)}
\label{tab:parameter_counts}
\small
\begin{tabular}{lcc}
\toprule
\textbf{Dataset} & \textbf{SCVI (M)} & \textbf{Geneformer (M)} \\
\midrule
Larry & 64.9 & 8.3 \\
MERFISH & 1.3 & 1.9 \\
PBMC & 53.2 & 7.1 \\
Shendure & 116.8 & 13.5 \\
\bottomrule
\end{tabular}
\end{table}

\subsection{MobileNet for image classification}
We finetune the ImageNet pretrained MobileNetV3 architecture distributed with PyTorch \cite{Paszke2019-bf}. To adapt it to our 5-class subset of Caltech101, the final classification layer is replaced with a fully connected layer with a 5-dimensional output. To adapt it to the 8-class kidney cortex dataset, we similarly replace the classification layer with a 8-dimensional output. For Caltech101, we use a $1:1$ train-test split, and optimize a cross-entropy loss for 5 epochs. For the kidney cortex dataset, we use the train-test split provided with MedMNIST \cite{Yang2023-eg} (165,466 and 47,280 images respectively), and train for 30 epochs. Training and implementation details for the kidney cortex dataset experiments are provided in Table~\ref{tab:app_mobilenet_details}.

\begin{table}[h]
\centering
\caption{Hyperparameters and Implementation Details for Kidney Cortex Cell Type Annotation with MobileNetV3}
\begin{tabular}{ll}
\toprule
\textbf{Parameter} & \textbf{Value} \\
\midrule
\multicolumn{2}{l}{\textit{Model Architecture}} \\
Base model & MobileNetV3-Small \\
Weight initialization & ImageNet pretrained (\verb|IMAGENET1K_V1|) \\
Input Channels & 3 (grayscale converted by repeating) \\
Image Size & 224 $\times$ 224 \\
\midrule
\multicolumn{2}{l}{\textit{Training Hyperparameters}} \\
Optimizer & Adam \\
Learning rate & $1 \times 10^{-3}$ \\
Finetuning objective & CrossEntropyLoss \\
Batch size & 512 \\
Epochs & 30 \\
\midrule
\multicolumn{2}{l}{\textit{Embedding Extraction}} \\
Representation & Last layer before classifier \\
\midrule
\multicolumn{2}{l}{\textit{Information probing}} \\
Method & Latent mutual information with 16 latent dimensions \\
Auxiliary signal & Class label (either one-vs.-all binary, or one-hot 8-way) \\
\bottomrule
\end{tabular}
\label{tab:app_mobilenet_details}
\end{table}

\subsection{ESM2 for SARS-CoV2 spike protein sequence representations}

We obtained pretrained ESM2 models of three different sizes (8M, 35M, 150M) as distributed with the HuggingFace \verb|transformers| library \cite{Wolf2020-jb}. We initialize models at the pretrained weights, and train them for a single epoch on the collection of SARS-CoV2 protein sequences curated from GISAID \cite{Shu2017-au} (preprocessed as described in \textbf{Supplemental Text~\ref{app:details_data}}). Training details including hyperparameters are summarized in \textbf{Supplemental Table~\ref{tab:app_esm_details}}.

\begin{table}[h]
\centering
\caption{Hyperparameters and Implementation Details for Sequence Experiments with ESM}
\begin{tabular}{ll}
\toprule
\textbf{Parameter} & \textbf{Value} \\
\midrule
\multicolumn{2}{l}{\textit{Model Architecture}} \\
Base models & ESM2 (8M, 35M, 150M parameters) \\
Finetuning objective & Masked Language Modeling (15\% masking) \\
Max sequence length & 1024 \\
\midrule
\multicolumn{2}{l}{\textit{Hyperparameters}} \\
Optimizer & AdamW \\
Learning rate & $1 \times 10^{-4}$ \\
Batch size & 8 \\
Epochs & 1  \\
\midrule
\multicolumn{2}{l}{\textit{Embedding Extraction}} \\
Pooling method & Mean pooling of last hidden state \\
\midrule
\multicolumn{2}{l}{\textit{Information Probing}} \\
Method & Latent mutual information with 16 latent dimensions \\
Auxiliary signal & Collection month (months since 01/2020) \\
\bottomrule
\end{tabular}
\label{tab:app_esm_details}
\end{table}
\section{Additional details on cell number scaling behavior}\label{app:cellnumberdetails}




Because the focus of this paper is not on sample number scaling, we confine to this supplemental note a discussion on the differences we observed between learning rates and learning capacity of the modeling approaches with respect to cell number. From the scaling relationship discussed in the main text, the saturating performance, $I_\infty$, measures the capacity of a model to capture auxiliary information in the limit of infinite data. The saturation scale parameter $N_{\text{sat}}$ quantifies the number of cells required to approach saturation, specifically, to be within 1 bit of $I_\infty$. The scaling exponent $s$ describes the model's sensitivity to new data (when $N\lesssim N_\textrm{sat})$. The estimated parameters for each model and task are shown in \textbf{Supplementary Fig.~\ref{fig:cellnumberparams}}.

scVI consistently showed the highest saturating performance $I_\infty$, suggesting it learns the highest quality representations given sufficient data. The representations learned by Geneformer are less informative, with discrepancies ranging from 0.6 to 1.1 bits across tasks compared to scVI -- corresponding to approximately halving the complexity of the captured signal. 

PCA alone was competitive with Geneformer for some tasks: it showed lower saturating performance than Geneformer in capturing protein abundance and developmental time, but surpassed Geneformer in capturing clonal and spatial information. This suggests that Geneformer is not well-suited to these tasks. 

The models also differ considerably in their saturation scale. PCA learns with very little data -- with $N_{\text{sat}}$ in the tens of cells for all metrics. This indicates that PCA representation quality saturates almost immediately. This rapid convergence is expected of linear models \cite{Cai2012-me} and suggests that PCA representations do not benefit from large dataset sizes. While Geneformer and scVI saturate much less rapidly, most models and metrics still have an $N_{\text{sat}} < 10^4$ cells. As this is lower than the actual dataset sizes, it suggests model performance on the evaluated metrics is largely saturated with respect to cell number.

Finally, the scaling exponent $s$ describes the model's sensitivity to dataset size prior to saturation. PCA consistently demonstrates the largest $s$ across all tasks. This indicates a high initial sensitivity to cell number, but this steep improvement reaches saturation quickly, as shown by $N_\text{sat}$.

A related work, \citet{DenAdel2024-ue}, studies model performance as a function of pretraining dataset size across a different set of tasks and metrics. \citet{DenAdel2024-ue} compute a ``learning saturation point'', the dataset size at which 95\% of the maximum performance is achieved with respect to tasks like cell type annotation and batch integration. While not quantitatively comparable to a scaling law parameter, this metric describes the same saturation phenomenon as $N_{\text{sat}}$ in this work. The results of both works are qualitatively similar: all models have early saturation points, and linear models have earlier saturation points than deep-learning approaches.

\begin{figure}
    \centering
    \includegraphics[width=0.9\linewidth]{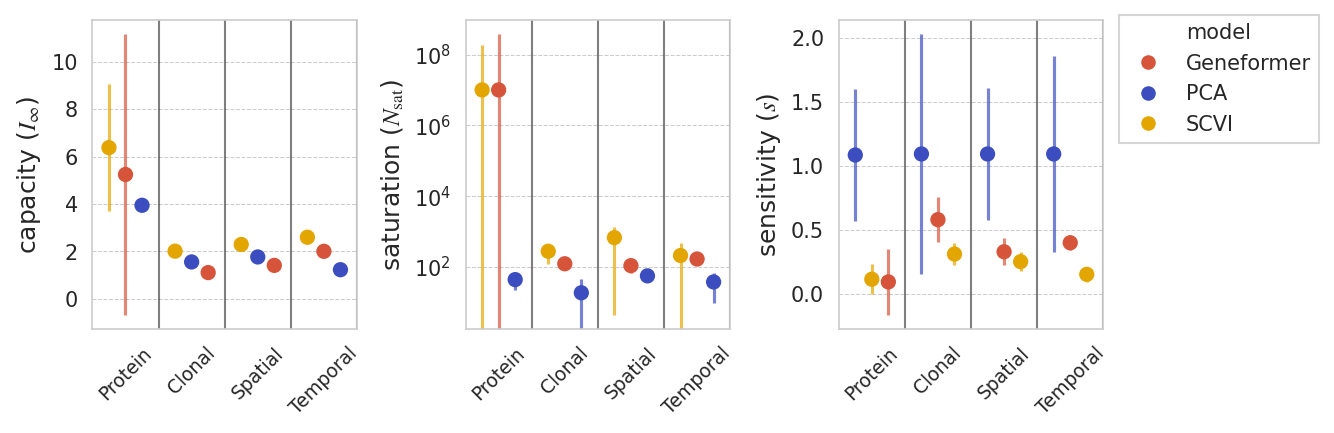}
    \caption{Comparison of cell number scaling parameters across model families and representation quality metrics. Error bars denote $2\sigma$ confidence interval. Random projection is omitted due to lack of scaling behavior. All parameters are estimated from datasets without artificial noise.}
    \label{fig:cellnumberparams}
\end{figure}

\section{Analytical results for a toy model of noise propagation}\label{app:proof}

Let $X, Y$ be multivariate Gaussian random vectors representing signals distributed as follows
$$Y \sim \mathcal{N}(\mu_Y, \Sigma_Y)$$
$$X = Y + U$$
where $U \sim \mathcal{N}(0, \Sigma_U)$.

Next, let $Z$ be a random vector representing a noisy measurement of $X$ with signal-to-noise ratio $\eta$:
$$Z = \sqrt{\eta} X + \mathcal{N}(0, I_n)$$ 

In our empirical results for transcriptomic data, $X$ corresponds to the true transcript counts, $Y$ corresponds to an auxiliary signal, and $Z$ corresponds to the representation extracted from a noisy measurement of $X$. We are interested in how $I(Y; Z)$ scales as a function of $\eta$. We next show that in the above toy model, the relationship between $\eta$ and $I(Y;Z)$ can be exactly specified.

\begin{theorem}[Theorem 3.1]
For the three variable Gaussian noise model specified above,

\begin{equation}
I(Y; Z) = \frac{1}{2} \log \frac{\det(\Sigma_Y + \Sigma_U + \eta^{-1}I_n)}{\det(\Sigma_U + \eta^{-1}I_n)}    
\end{equation}

In the special case where $n = 1$, denoting the variances $\sigma^2_Y, \sigma^2_U$:

\begin{equation}\label{main}
I(Y; Z) = \frac{1}{2} \log \frac{\eta (\sigma^2_Y + \sigma_U^2) + 1}{1 + \sigma_U^2 \eta}    
\end{equation}

\end{theorem}

\begin{proof}

We build on a result for Gaussian vector noise channels (see \cite{Polyanskiy2024-nx, Guo2004-yr}) which states that for independent Gaussian random vectors $X, N$, 
$$I(X; X+N) = \frac{1}{2} \log \frac{\det (\Sigma_x + \Sigma_N)}{\det(\Sigma_N)}$$
where $\Sigma_X, \Sigma_N$ are the covariance matrices of $X, N$.

We will begin by rewriting $Z$ in terms of $Y$. From definitions, we have
$$Z = \sqrt\eta X + \mathcal{N}(0, I_n)$$
$$= \sqrt\eta (Y + U) + \mathcal{N}(0, I_n)$$
$$= \sqrt\eta (Y + \mathcal{N}(0, \Sigma_U)) + \mathcal{N}(0, I_n)$$

Due to closure rules for Gaussians, we can rewrite
$$Z = \sqrt\eta Y + \mathcal{N}(0, \eta^2\Sigma_U + I_n)$$

Next, we observe that due to the scale invariance of mutual information \cite{Thomas-M-Cover2006-xh}
$$I(Y; Z) = I(Y; \eta^{-1/2} Z)$$
$$= I(Y; Y + \eta^{-1/2} \mathcal{N}(0, \eta \Sigma_U +  I_n))$$
$$= I(Y; Y + \mathcal{N}(0, \Sigma_U + \eta^{-1}I_n))$$

Now we can directly apply the Gaussian vector channel result:
$$I(Y;Z) = I(Y; Y + \mathcal{N}(0, \Sigma_U + \eta^{-1}I_n)) $$
$$= \frac{1}{2} \log \frac{\det(\Sigma_Y + \Sigma_U + \eta^{-1}I_n)}{\det(\Sigma_U + \eta^{-1}I_n)}$$

And in the special case where $n=1$, we have that
$$I(Y; Z) = \frac{1}{2} \log \frac{\sigma_Y^2 + \sigma_U^2 + \eta^{-1}}{\sigma_U^2 + \eta^{-1}} $$
$$= \frac{1}{2} \log \frac{\eta \sigma^2_Y + \eta \sigma_U^2 + 1}{1 + \sigma_U^2 \eta}$$
    
\end{proof}

\section{Relating molecular measurement sensitivity to signal-to-noise ratio}\label{app:umi-to-snr}

In our analysis of transcriptomic data, we relate the SNR to the mean UMI per cell as $\eta \propto \textrm{UMI}$. Here, we provide a statistical motivation for this relationship.

The definition of SNR that we introduce in Eq.~(1) is the power ratio, (as in \cite{Guo2004-yr, Thomas-M-Cover2006-xh}):

\begin{equation}
\text{SNR} := \frac{\mathbb{E} [S^2]}{\mathbb{E} [N^2]}    
\end{equation}

where $S$ is a signal corrupted by additive noise $N$.

This definition aligns with the choice of $\text{SNR} = \eta$ in the scalar case of the noise propagation model introduced in \textbf{Box \ref{box:theory}} and \textbf{Supplemental Text \ref{app:proof}}. The observed noisy signal is defined as $Z = \sqrt{\eta} X + \mathcal{N}(0, 1)$, so the power-ratio is

\begin{equation}
    \text{SNR}_{\text{model}} = \mathbb{E} [(\sqrt\eta X)^2] = \eta \mathbb{E} [X^2]
\end{equation}

where $\eta$ is equivalent to the power-ratio notion of SNR when $X$ is normalized such that $\mathbb{E} [X^2] = 1$. This is the convention followed in \citet{Guo2004-yr}.

In the transcriptomic data, a measurement is a number of molecules detected in a cell, following $u \sim \text{Pois} (\lambda)$. We define a signal component as the mean $\lambda$ and a noise component as fluctuations $u - \lambda$, which together comprise the noisy measurement. Now, we can define a power-ratio SNR for counts as

\begin{equation}
\text{SNR}_{\text{Pois}} = \frac{\lambda^2}{\mathbb{E}[(u - \lambda)^2]} = \text{CV}^{-2} \propto \lambda    
\end{equation}

As such, we take $\eta \propto \lambda$ where $\lambda$ is the mean UMI count per cell.

\section{Inverting the noise scaling form}\label{sec:app_UMI90}

Measurement noise scaling laws can be used to determine the data quality or sample quantity necessary to obtain a model with a specified amount of information. The fit parameter $\bar{\eta}$ relates to the measurement sensitivity necessary to reach at least $\frac{1}{2}$ bits below $\mathcal{I}_{\max}$. More generally, one can invert Eq. \ref{eqn:scaling_general} to obtain a function, $\eta(\mathcal{I})$, which estimates the acceptable SNR necessary to learn a representation with a given information content with respect to a specified external signal. For the case of transcriptomic data, where $\text{UMI}\propto \eta$,

\begin{equation}\label{eqn:umi90}
   \text{UMI}(\mathcal{I}) = \bar{\eta}\frac{2^{2\mathcal{I}} - 1}{ 2^{2 \mathcal{I}_{\max}} - 2^{2\mathcal{I}}}
\end{equation}

Using Eqn.~\ref{eqn:umi90}, we compute the UMI90 required to saturate auxiliary-MI across all the model-task pairs studied here, and provide these in \textbf{Supplemental Table~\ref{tab:u90}}.

\begin{table}[h]
\centering
\begin{tabular}{lrrrrr}
\hline
Metric & Geneformer & PCA & Random Projection & SCVI & Actual UMIs \\
\hline
Temporal MI & $2571.7 \pm 201.1$ & $5606.4 \pm 994.0$ & $8012.5 \pm 1028.9$ & $2815.5 \pm 198.8$ & 2500 \\
Clonal MI & $4212.9 \pm 854.2$ & $20190.0 \pm 3764.7$ & $7339.6 \pm 1673.0$ & $5073.3 \pm 457.0$ & 2580 \\
Spatial MI & $2212.0 \pm 537.3$ & $32107.5 \pm 31922.5$ & $1339.9 \pm 177.0$ & $5000.4 \pm 830.7$ & 367 \\
Protein MI & $8090.1 \pm 1510.3$ & $46266.6 \pm 9330.5$ & $46989.6 \pm 5007.8$ & $44958.5 \pm 6071.7$ & 8100 \\

\hline
\end{tabular}
\caption{$\eta_{90}$ values by metric and model family, with $\pm2\sigma$.}
\label{tab:u90}
\end{table}

\clearpage

\end{appendices}

\bibliography{sn-bibliography}
\end{document}